\documentclass[a4paper,10pt]{article}
 
\usepackage{amsmath,amssymb,amsthm}
\usepackage{xspace}
\usepackage{mathtools}
\mathtoolsset{showonlyrefs=true}
\usepackage[english]{babel}
\usepackage{color}
\usepackage[T1]{fontenc}
\usepackage[utf8]{inputenc}
\usepackage{constants}
\newconstantfamily{small}{symbol=c,}
\usepackage[margin=1in]{geometry}
\usepackage{bm}
\usepackage{float}
\floatstyle{boxed}
\newfloat{pseudocode}{thb}{pseudo}
\floatname{pseudocode}{Pseudocode}
\usepackage{hyperref}

\usepackage{accents}
\usepackage{botex}

\newtheorem{definition}{Definition}[section]
\newtheorem{lemma}[definition]{Lemma}
\newtheorem{theorem}[definition]{Theorem}
\newtheorem{cor}[definition]{Corollary}
\newtheorem{proposition}[definition]{Proposition}
\newtheorem{observation}{Observation}
\newtheorem{remark}{Remark}
\newtheorem{fact}[observation]{Fact}

\newcommand{\Prob}[2]{\mathbf{P}_{#1} \left( #2 \right)}

\newcommand{\graph}{\mathcal{G}}
\newcommand{\nodes}{\mathcal{V}}
\newcommand{\edges}{\mathcal{E}}
\newcommand{\source}{\mathrm{source}}
\newcommand{\sink}{\mathrm{sink}}
\newcommand{\w}{w}
\newcommand{\pot}[1]{{p}^{(#1)}}
\newcommand{\vpot}[1]{\mathbf{p}^{(#1)}}
\newcommand{\est}[2]{{V}_{#2}^{(#1)}}
\newcommand{\ntok}[2]{{Z}_{#2}^{(#1)}}
\newcommand{\vol}{\mathrm{vol}}
\newcommand{\volmin}{\mathrm{vol}_\mathrm{min}}
\newcommand{\volmax}{\mathrm{vol}_\mathrm{max}}
\newcommand{\bone}{\mathbf{b}}

\newcommand{\isthere}[1]{{X}^{(#1)}}

\newcommand{\eigr}{\underline{\lambda}}
\newcommand{\reig}{\overline{\lambda}}
\newcommand{\modA}{\overline{A}}
\newcommand{\modW}{\underline{P}}
\newcommand{\modN}{\underline{N}}
\newcommand{\modD}{\underline{D}}
\newcommand{\gA}{\underline{A}}
\newcommand{\gD}{\overline{D}}
\newcommand{\modG}{\overline{\mathcal{G}}}

\newcommand{\modlap}{\underline{L}}
\newcommand{\glap}{\overline{L}}
\newcommand{\nlap}{\mathcal{L}}
\newcommand{\nglap}{\underline{\nlap}}
\newcommand{\smat}{\ensuremath{\Delta}}
\newcommand{\modb}{\underline{\vec b}}

\newcommand{\energy}{{E}}
\newcommand{\modx}{\underline{\vec x}}
\newcommand{\grM}{\mathrm{gr}(M)}
\newcommand{\grP}{\mathrm{gr}(\trans)}
\newcommand{\grmodP}{\mathrm{gr}(\modW)}

\newcommand{\trans}{P}

\newcommand{\eigtrans}{\rho}
\newcommand{\grho}{\underline{\rho}}
\newcommand{\glam}{\underline{\lambda}}

\usepackage{xpatch}
\xpretocmd{\footnote}{\unskip}{}{}

\usepackage{mathtools}
\usepackage{commentstodos}

\def\bx{{\bf x}}
\def\by{{\bf y}}

\def\bu{{\bf u}}
\def\bw{{\bf w}}
\def\bz{{\bf z}}

\allowdisplaybreaks

\usepackage{graphics}
\usepackage{bbm}

\makeatother

\usepackage[boxed,linesnumbered]{algorithm2e}
\newcommand\eqdef{\defas}

\renewcommand{\bone}{{\vec 1}}

\global\long\def\error{\vec e}
\global\long\def\errort{\error_{\perp}}
\global\long\def\parcoef{\alpha}

\newcommand{\inorm}[1]{\Vert #1\Vert} % inline version

\begin{document}

\title{\textbf{Pooling or Sampling: Collective Dynamics\\ for 
Electrical Flow Estimation}
\thanks{This work is partly supported by the EU FET project MULTIPLEX no. 
317532 and by the National Science Foundation under Grants No. CCF 1540685 and
CCF 1655215.}
}
\author{Luca Becchetti\\
    {\footnotesize{}Sapienza Università di Roma}\\
    {\footnotesize{} Rome, Italy}\\
    {\footnotesize{}\texttt{becchetti@diag.uniroma1.it}} 
    \and Vincenzo Bonifaci \\ 
    {\footnotesize{}Consiglio Nazionale delle Ricerche}\\
    {\footnotesize{} Rome, Italy}\\
    {\footnotesize{}\texttt{vincenzo.bonifaci@iasi.cnr.it}} 
    \and Emanuele Natale\\
    {\footnotesize{}Max Planck Institute for Informatics}\\
    {\footnotesize{}Saarbr\"ucken, Germany }\\
    {\footnotesize{}\texttt{emanuele.natale@mpi-inf.mpg.de}} 
}

\maketitle

\begin{abstract}  % put your abstract here!
    The computation of electrical flows is a crucial primitive for many
    recently proposed optimization algorithms on weighted networks. 
    While typically implemented as a centralized subroutine, the 
    ability to perform this task in a fully 
    decentralized way is implicit in a number of biological systems. 
    Thus, a natural question is whether this task can provably be 
    accomplished in an efficient way by a network of agents executing a simple protocol. 
    
    We provide a positive answer, proposing two distributed approaches 
    to electrical flow computation on a weighted network: a 
    deterministic process mimicking Jacobi's
    iterative method for solving linear systems, and a randomized token
    diffusion process, based on revisiting a classical random
    walk process on a graph with an absorbing node. 
    We show that both processes converge to a solution of Kirchhoff's node
    potential equations, derive bounds on their convergence rates in terms 
    of the weights of the network, and analyze their time and
    message complexity. 
    
        %The computation of electrical flows is a crucial primitive for many
    %recently proposed optimization algorithms on weighted networks. However,
    %this primitive is typically implemented as a centralized subroutine. We
    %consider two distributed approaches to compute electrical flows on a
    %weighted network: a deterministic distributed process based on Jacobi's
    %iterative method for solving linear systems, and a randomized token
    %diffusion process, based on a new exploitation of a classical random
    %walks process on a graph with an absorbing node. 

    %We show that both processes converge to a solution of Kirchhoff's node
    %potential equations, derive bounds on the convergence rates of the two
    %processes in terms of the weights of the network, and analyze their time and
    %message complexity. 
    %%
    %As our token diffusion process is the first fully-local estimation method
    %for electrical flow, we believe that it could be efficiently combined with
    %multi-agent-based algorithms which adaptively exploit the calculation of
    %electrical flow, to obtain novel dynamic optimization algorithms.
\end{abstract}

%%%%%%%%%%%%%%%%%%%%%%%%%%%%%%%%%%%%%%%%%%%%%%%%%%%%%%%%%%%%%%%%%%%%%%%%%%%%%%%%%%%%%%%%%%%%%%%%%%%%%%%%%
%% start of main body of paper

\section{Introduction}
The computation of currents and voltages in a resistive electrical network, besides being an  interesting problem on its own, is a crucial primitive in many recently proposed optimization algorithms on weighted networks. Examples include the fast computation of maximum flows \cite{Christiano:2011,Lee:2013}, network sparsification \cite{Spielman:2011}, and the generation of random spanning trees \cite{Kelner:2009}. 

Solving the electrical flow problem requires solving a system of linear equations, whose variables are the electrical voltages, or ``potentials'', at the nodes of the network (equivalently, the currents traversing its edges). Performing this task can be computationally nontrivial, and is typically achieved in a centralized fashion. 

At the same time, the ability to perform this task in a fully decentralized way is implicit in a number of biological systems by virtue of the electronic-hydraulic analogy \cite{kirby_micro-_2013}, including the \emph{P. polycephalum} slime mold \cite{Tero:2007, bonifaci2012physarum,Straszak:2016:soda} and ant colonies \cite{Ma:2013}. 
These organisms have been showed to implicitly solve the electrical flow problem in the
process of forming food-transportation networks.
Such capability of biological systems naturally raises the following questions, which motivate our paper:
\begin{enumerate}
\item[(Q1)] Can this task be collectively accomplished by the network itself, if every node is an agent that follows an elementary protocol, and each agent can only interact with its immediate neighbours, otherwise possessing no knowledge of the underlying topology?
\item[(Q2)] In case of a positive answer to {\em Q1}, what is the involved computational effort for the network, in terms of convergence time and communication overhead?
\end{enumerate}

We address the two aforementioned questions by providing analytical bounds
which are of interest for many bio-inspired multi-agent systems in swarm
robotics and sensor networks (e.g.
\cite{russell_ant_1999,fujisawa_designing_2014}).

\subsection{Our contribution}
We propose two complementary, fully 
decentralized approaches to electrical flow computation on a weighted 
network. In particular, Towards question (Q1), we make the following 
contributions:
\begin{enumerate}
\item We consider a deterministic distributed 
process, based on \emph{Jacobi's iterative method} for solving linear 
systems. This process converges to a solution of Kirchhoff's node 
potential equations. We bound the convergence rate of this process in 
terms of a graph-theoretic parameter of the network -- \emph{graph 
conductance}. 

\item Driven by a natural probabilistic interpretation of the 
aforementioned process, we further consider a \emph{randomized token 
diffusion} process, implementing Monte Carlo sampling via independent 
random walks. This process also 
converges to a solution of Kirchhoff's node potential equations, but
differently from the deterministic algorithm, randomized token 
diffusion does not involve any arithmetics on 
real numbers: each agent/node simply maintains a counter of the number of 
random walks currently visiting it, from which a simple unbiased 
estimator of the node's potential can be derived.
We derive a bound on the convergence rate of this process in terms of another 
graph-theoretic parameter -- \emph{edge expansion}. 
\end{enumerate}

%\smallskip
With respect to question (Q2), while the strong connection between electrical flows and
random walks has been known for a while and has been extensively 
investigated in the past \cite{Doyle:1984,Tetali:1991}, any effective exploitation of 
electrical flow computation crucially requires explicit and plausible bounds on the 
efficiency and accuracy of the algorithm(s) under consideration. 
In this respect, besides establishing the correctness of the two algorithmic approaches
above, our core contribution is to derive detailed bounds on their time and
communication complexities in terms of fundamental combinatorial properties of
the network.

%\smallskip
Finally, our results highlight the algorithmic potential of classical 
models of opinion formation, as discussed in more detail in the next 
section.

\subsection{Related work}
We briefly review contributions that are most closely related to the spirit of 
this work.
\paragraph{Computing electrical flows.}
The problem of computing voltages and currents of a given resistive
network, that is, the question of solving Kirchhoff's equations, is a classical
example of solving a linear equation system with a \emph{Laplacian} constraint
matrix \cite{Vishnoi:2013}. While Jacobi's method is a well-known 
approach to the solution of a class of linear systems that subsumes Laplacian systems, its
complexity analysis in the literature (for example, in
\cite{Saad:2003,Varga:2010}) is generic, and does not exploit the additional
matrix structure that is inherent in Laplacian systems. In our setting, 
existing results on the convergence of Jacobi's method could be leveraged to
prove convergence to a correct solution, but they would fail to provide 
explicit bounds on convergence rate. 

\paragraph{Electrical flows and random walks.}
Relations between electrical quantities and statistical properties of
random walks have been known for a long time, and are nicely discussed, for
example, in Doyle and Snell \cite{Doyle:1984}, Lovasz \cite{Lovasz:1996}, and Levin, Peres and Wilmer \cite{Levin:2009}. In particular,
in their monograph, Doyle and Snell point out the interpretation of the electrical current along an
edge as the \emph{expected number of net traversals} of the edge by a random
walker that is injected at the source node and absorbed at the sink node of the
flow. 
While our randomized token diffusion algorithm refers to the same underlying
process, it crucially differs from the former in the  
interpretation of electrical current, as it links the electric 
potential of a node to the expected number of random walks currently 
visting the node (see Section \ref{se:rw}%Lemma \ref{le:estimator}
). 

This connection has been also explored by Chandra et al. \cite{Chandra:1997}, who
characterized the \emph{cover time} of the random walk in terms of the maximal
effective resistance of the network, as well as by Tetali \cite{Tetali:1991},
who characterized the \emph{hitting time} of the random walk in terms of the
effective resistance between source and sink. Tetali \cite{Tetali:1991} also
proved that the expected number of visits to a node by a random walker injected
at the source and absorbed at the sink is related to the electrical potential of
the node in a simple way. In principle, like the
one discussed by Doyle and Snell, this characterization could be
used as the  basis for another random walk-based approach to the estimation of
electric potentials, with essentially the same complexity as the method we
propose. Still, this is a static characterization that, by itself,
does not provide an iterative algorithm or error bound. 
On the contrary, our interpretation results in an
estimator, which depends solely on the number of tokens at each node and 
is thus entirely \emph{local}, as opposed to previous methods, which 
entail tracking an event that depends on \emph{global} properties of the network (such as the hitting
time of a specific node or the absorption at the sink).
Hence, we believe our proposed randomized diffusion process is more
suitable to accommodate dynamic changes in the weights of the network when
coupled with other processes, such as the \emph{Physarum} dynamics
\cite{bonifaci2012physarum, Ma:2013}. 

\paragraph{Electrical flows and complex systems.}
Understanding how electrical flows are
computed in a decentralized fashion can help explain the emergent 
behavior of certain social and biological systems. 
For example, foraging behaviors of the \emph{P. polycephalum} slime
mold \cite{Tero:2007,bonifaci2012physarum,Straszak:2016:soda}
and of ant colonies \cite{Ma:2013} can both be formulated in terms of
\emph{current-reinforced} random walks, see Ma et al. \cite{Ma:2013}. In this
respect, our results can be seen as a step towards a more thorough
understanding of these complex biological processes, \emph{at the microscopic scale}. 
Moreover, the simple processes we propose shed new light on the computational 
properties of models of opinion dynamics in social networks
\cite{degroot1974reaching,acemoglu2011opinion,mossel_majority_2014}. In 
particular, the classical model of opinion formation 
proposed by DeGroot \cite{degroot1974reaching} essentially corresponds to 
the decentralized version of Jacobi's iterative method presented in 
this paper.\footnote{The only difference is the presence of two ``special'' 
agents (the source and the sink), whose behaviours slightly differ 
from the others, in that they exchange information with the exterior in the 
form of a current flow.} This is a hint that opinion dynamics are 
extremely versatile processes, whose algorithmic potential is not 
completely understood.

\paragraph{Distributed optimization.}
Since the electrical flow is one of minimum energy, the
decentralized computation of electrical flows can be seen as an instance of
\emph{distributed optimization}, 
akin to the problems considered within the multiagent framework introduced by 
\cite{Tsui:2003} and of potential interest for the class of distributed
constraint optimization problems considered, for example, in \cite{Modi:2003}. 
Although for other distributed problems it has been suggested that Laplacian-based approaches are not the most computationally effective \cite{Elhage:2010}, as we have mentioned in the previous paragraph current-reinforced random walks are considered a feasible model, at least for certain biological systems \cite{Ma:2013,Figueiredo:2017}. 
In this paper, we leveraged the specific structure of electrical flows to
prove the effectiveness of our decentralized solutions. In particular, the
dependency of convergence rates on the size of the network is polynomial,
which cannot be claimed for other more generic approaches to distributed optimization. 

Finally, in the rather different context of social choice and agents with
preferences (as opposed to our perspective motivated by natural processes for
network optimization),
\cite{salehi-abari_empathetic_2014} investigate mechanisms for social choice on
social networks as a weighted form of classical preference aggregation. One of
the update processes they consider is loosely related  to our Jacobi process.

\subsection{Outline}
The rest of this paper is organized as follows. In Section \ref{se:prelim} we
discuss some preliminaries about electrical networks and flows and set up the
necessary notation and terminology. In Section \ref{se:jacobi} we describe and
analyze the deterministic distributed algorithm based on Jacobi's method for the
solution of Kirchhoff's equations. In Section \ref{se:rw} we propose and analyze
our randomized token-diffusion method for the estimation of the electric
potentials. We conclude by summarizing our findings in Section \ref{se:conclu}.

\section{Preliminaries on electrical networks and notation}
\label{se:prelim}

We consider a graph $\graph=(\nodes,\edges,w)$, with node set $\nodes$, edge set $\edges$, and positive edge weights $(w_e)_{e \in \edges}$ representing electrical conductances. We also denote the weight of an edge by $w_{uv}$ if $u, v \in \nodes$ are the endpoints of the edge; if no edge corresponds to the pair $(u,v)$, $w_{uv}=0$. We use $n$ and $m$ to denote the number of nodes and edges, respectively, of the graph. Without loss of generality we assume that $\nodes=\{1,2,\ldots,n\}$. 

The (weighted) \emph{adjacency matrix} of $\graph$ is the matrix $A \in
\Real^{n \times n}$ whose $(u,v)$-entry is equal to $w_{uv}$ if $\{u,v\} \in
\edges$, and to $0$ otherwise. 
The \emph{volume} (or generalized degree) of a node $v$ is the total weight of
the edges incident to it, and is denoted by $\vol(v)$. The generalized
\emph{degree matrix} $D \in \Real^{n\times n}$ is the diagonal matrix with
$D_{uu} = \vol(u)$. 
The matrix $P \defas D^{-1}A$ is the \emph{transition matrix} of $\graph$; we denote its eigenvalues (which are all real) by 
$\eigtrans_1 \ge \ldots \ge \eigtrans_n$. 
We use $\volmin$ and $\volmax$ to denote the smallest and largest volume,
respectively, of the nodes of $\graph$.

We adhere to standard linear algebra notation, and we reserve boldface type for vectors. We denote by $\vec \chi_i$ the $i$-th standard basis vector, that is, a vector whose entries are 0 except for the $i$-th entry which is $1$. With $\vec 0$ and $\vec 1$ we denote vectors with entries all equal to 0 and 1, respectively. 

In the next sections, we make use of the following fact. 
\begin{fact}
    \label{fact:similar}
    The transition matrix $\trans \defas D^{-1}A$ is similar to the symmetric matrix
    $N \defas D^{-1/2}AD^{-1/2}$ via conjugation by the matrix $D^{-1/2}$. In
    particular, we have
    \[
        \trans^{t} 
            = (D^{-1} A)^{t} 
            = (D^{-\frac 12} N D^{\frac 12})^{t} 
            = D^{-\frac 12} N^{t} D^{\frac 12}. 
    \]
    Moreover, thanks to the fact that $N$ is symmetric, $N$ has $n$ orthonormal
    eigenvectors $\vec x_1, \dots, \vec x_n$, which correspond to the eigenvectors
    $\vec y_1, \dots, \vec y_n$ of $\trans$ via the similarity transformation
    $\vec x_i = D^{1/2} \vec y_i$ for each $i$. 
    Observe also that both $\vec x_i$ and $\vec y_i$, for each $i$, are associated to
    the same eigenvalue $\eigtrans_i$ of $\trans$.
\end{fact}

The \emph{Laplacian matrix} of $\graph$ is the $n \times n$ matrix $L \defas D - A$. 
%WRONG!
%Similarly to Fact \ref{fact:similar},
%since $L = D(I-P)$ we have that $L$ is similar to $I-P$. Thus, $L$'s
%eigenvalues $\lambda_1,\cdots,\lambda_n$ equal $1-\rho_1, \cdots, 1-\rho_n$,
%respectively.
We denote by $\Lambda_1 \le \ldots \le \Lambda_n$ the eigenvalues of $L$. We often use the facts that $\Lambda_1 = 0$ and that $L \cdot \vec 1 = \vec 0$. 

In our setting, one node of the graph acts as the \emph{source}, and one as the
\emph{sink} of the electrical flow. Kirchhoff's equations for a network
$\graph$ are then neatly expressed by the linear system 
\[ 
    L \vec p = \vec b, 
\]
where $\vec p $ is the unknown vector
of electric \emph{potentials}, and $\vec b \in \Real^n$ is a vector such that
$b_{\source}=1$, $b_{\sink}=-1$, and $b_u = 0$ if $u \notin \{\source,
\sink\}$. 
The \emph{electrical flow} is easily obtained from the vector $\vec p$: the electrical flow along an edge $\{u,v\}$, in the direction from $u$ to $v$, equals $w_{uv} \cdot (p_u - p_v)$. 

For a given weighted graph with a source and a sink, the electrical flow is uniquely defined. 
We remark however that Kirchhoff's equations have infinite solutions, since electric  potentials are defined up to any constant offset: if $L \vec p = \vec b$, then $L (\vec p + c \bone) = \vec b$ for any constant $c$ (since $L \bone = \vec 0$.) We call the (unique) solution $\vec p$ such that $p_{\sink}=0$ the \emph{grounded} solution to Kirchhoff's equations. 

The \emph{graph conductance} (or \emph{bottleneck ratio}) of graph $\graph$ is the constant
\[
\phi(\graph) \defas \min_{S \subset \nodes \,:\, \vol(S) \le \vol(\nodes)/2} 
\frac{ w(S, \nodes \setminus S) }{\vol(S)}, 
\]
where $\vol(S)$ denotes the total volume of the nodes in $S$, and $w(S, \nodes\setminus S)$ denotes the total weight of the edges crossing the cut $(S, \nodes\setminus S)$. 
The \emph{edge expansion} of $\graph$ is the constant
\[
    \theta(\graph) \defas \min_{S \subset \nodes \,:\, \card{S} \le n/2}
    \frac{w(S,\nodes\setminus S)}{\card{S}}. 
\]
The graph conductance is a number between 0 and 1, while the edge expansion is a number between 0 and $\volmin$.

\section{Jacobi's method}
\label{se:jacobi}

The potentials $\vec p$ are a solution of the linear system
\begin{equation}
L \vec p = \vec b.\label{eq:p_def}
\end{equation}
A classic parallel iterative algorithm for solving such a system is Jacobi's method \cite{Saad:2003,Varga:2010}, which goes as follows. 
%Let\footnote{Note that $A$ is usually used to denote the binary adjacency matrix of a graph. Here instead it is a weighted adjacency matrix with edge weights $w_{e}$ for $e\in E$.} $L=D-\AofL$ where $D$ is the diagonal matrix whose diagonal elements coincide with those of $L$ and $-\AofL$ is such that all its extra-diagonal elements coincide with $L$ while its diagonal is zero: in formulas,
%$D_{ij}=\kron\cdot L_{ij}$ and $-\AofL_{ij}=\left(1-\kron\right)\cdot L_{ij}$
%for each $i$ and $j$, where $\kron$ is Kronecker's delta. 
System \eqref{eq:p_def} can be rewritten as $D \vec p-A\vec p= \vec b$,
which is equivalent to 
\begin{equation}
\vec p=D^{-1}\left(A \vec p+ \vec b\right).\label{eq:p_recur_fix}
\end{equation}

The idea underlying Jacobi's method is to introduce the related linear recurrence system 
\begin{equation}
\label{eq:p_recur}
\tilde{\vec p}\left(t+1\right)=D^{-1}\left(A\tilde{\vec p}\left(t\right)+ \vec b\right).
\end{equation}
%We remark that the ``time'' index $\tau$ of the recurrence system is distinct from the time $t$ of the Physarum dynamics.

For any node $u \in \nodes$, (\ref{eq:p_recur}) becomes
\begin{equation}
\tilde{p}_{u}(t+1)=\frac{b_u + \sum_{v\sim u} w_{uv} \tilde{p}_{v}(t)}{\sum_{v\sim u}w_{uv}},\label{eq:p_iterate}
\end{equation}
where the sums range on all neighbors $v$ of $u$. Note that the denominator in \eqref{eq:p_iterate} equals $\vol(u)$. 

Indeed, from (\ref{eq:p_recur_fix}) we know that $\tilde{\vec p}$ is a fixed
point of (\ref{eq:p_iterate}) whenever $L \tilde{\vec p}= \vec b$. This
suggests running the following message-passing algorithm at every node $u \in
\nodes$ (Algorithm \ref{algo:jaco}). 
\begin{algorithm}
	\DontPrintSemicolon % Some LaTeX compilers require you to use \dontprintsemicolon instead
    \SetKwFor{ProbDo}{with prob}{do}{}
    \SetKwInput{Parameter}{Parameter}
	%\KwIn{$u$}
	%\KwOut{The largest element in the set}
	\Parameter{$u \in \nodes$}
	\While{true}{%
%	{\bf \{Step 1: Sending potentials\}}\;
	\tcp*[h]{Step 1: Send potentials}  \;
	\For{every neighbor $v$ of $u$} {
            send $\tilde{p}_u$ to $v$ \; 
	}
%	{\bf\{Step 2: Receiving potentials\}}\;
	\tcp*[h]{Step 2: Receive potentials}  \;
	\For{every neighbor $v$ of $u$} {
		receive $\tilde{p}_v$ from $v$ \; 
	}
	\tcp*[h]{Step 3: Update potentials}  \;
	%Update $\tilde{p}_u$ using \eqref{eq:p_iterate}
	$\tilde{p}_{u} = \frac{1}{\vol(u)} \left( b_u + \sum_{v\sim u} w_{uv} \tilde{p}_{v} \right)$
	}% end while
	\caption{Jacobi's method for solving $L\vec p=\vec b$.}
	\label{algo:jaco}
\end{algorithm}

Algorithm \ref{algo:jaco} does not specify an initial value for $\tilde{\vec p}$: any initial value can be used. There is also no explicit termination condition. Terminating the algorithm sooner or later has only an effect on the numerical error, as explained in the next subsection. 

\subsection{Correctness and rate of convergence}

To study the convergence of Algorithm \ref{algo:jaco}, fix any solution $\vec p$ of \eqref{eq:p_def}, and define the \emph{error at step $t$} as
\begin{equation}
\error\left(t\right) \defas
\vec p-\tilde{\vec p}\left(t\right)=\errort\left(t\right)+\parcoef(t)\cdot\bone,\label{eq:e_decomp}
\end{equation}
 where $\errort\left(t\right)$ is the projection of $\vec p-\tilde{\vec p}\left(t\right)$
on the subspace orthogonal to $\bone,$ i.e.,  
\[
\errort\left(t\right)=\left(I-\frac{1}{n}\bone\bone\tp \right)\error\left(t\right),
\]
while $\parcoef(t)\cdot\bone$ is the component of $\vec p-\tilde{\vec p}\left(t\right)$
parallel to $\bone$, i.e., 
\[
\alpha{(t) \, \cdot\,\bone}=\frac{1}{n}\bone\bone\tp\error\left(t\right).
\]

The reason we decompose $\error\left(t\right)$ as in (\ref{eq:e_decomp})
is that $\vec p$ is defined in (\ref{eq:p_def}) \emph{up
to translation along $\bone$}, since $\bone$ is in the kernel of
$L$: 
%Observe that the pseudoinverse $L^{+}$ in (\ref{eq:kirch_current})
%chooses a $p$ orthogonal to $\bone$, because $L^{+}$ has the same
%eigenvectors of $L$ and thus $\bone$ is in the kernel of $L^{+}$
%as well; however, 
any vector $\vec p+\beta \cdot\bone$ satisfies $L\left(\vec p+\beta \cdot\bone\right)=\vec b$
as well. Therefore, one has converged to a solution of (\ref{eq:p_def})
as soon as $\tilde{\vec p}=\vec p+ \beta \cdot\bone$ for any $\beta$, which
implies that in (\ref{eq:e_decomp}) we do not care about the value
of $\parcoef$, we only care about $\errort\left(t\right)$ having a
small norm. This becomes clear by plugging the decomposition (\ref{eq:e_decomp})
in (\ref{eq:p_recur}): 
\begin{align}
        \vec p-\errort\left(t+1\right)-\parcoef {(t+1)} \cdot\bone &
        \eqdef\tilde{\vec p}\left(t+1\right)\nonumber\\
         & =D^{-1}\left(A\tilde{\vec p}\left(t\right)+\vec b\right)\nonumber\\
         & =D^{-1}\left(A\left(\vec
         p-\errort\left(t\right)-\parcoef{(t)}\cdot\bone\right)+\vec
         b\right)\nonumber\\
         & \stackrel{\left(a\right)}{=}\vec p-\trans \errort\left(t\right)-\parcoef{(t)}\cdot\bone,
    \label{eq:trick}
\end{align}
where in $\left(a\right)$ we used that for any transition matrix $\trans \defas D^{-1}A$ it
holds $\trans \bone=\bone$ and (\ref{eq:p_recur_fix}). From \eqref{eq:trick} it
follows that {$\errort(t+1) = \trans \errort(t) - (\alpha(t+1) - \alpha(t))
\bone$}. 

By projecting, it also follows that 
\[
    \begin{split}
        \errort(t+1) &= \left(I - \frac{1}{n} \bone \bone\tp\right) \trans
            \errort(t)\\
        &= \left(I - \frac{1}{n} \bone \bone\tp\right) \trans
            \left(I - \frac{1}{n} \bone \bone\tp\right) \trans
            \errort(t-1)\\
            &\overset{(a)}{=} \left(I - \frac{1}{n} \bone \bone\tp\right) 
            \left(\trans - \frac{1}{n} \bone \bone\tp\right) \trans
            \errort(t-1)\\
        &\overset{(b)}{=} \left(I - \frac{1}{n} \bone \bone\tp\right) 
            \trans^2
            \errort(t-1),
    \end{split}
\]
where in $(a)$ we used again that $\trans\bone=\bone$ and in $(b)$ we used that
$\left(I - \frac{1}{n} \bone \bone\tp\right)\frac{1}{n} \bone \bone\tp = 0$.
By repeating steps $(a)$ and $(b)$ above we can unroll the previous equation
and get  
\begin{equation}
    \errort(t) = \left(I - \frac{1}{n} \bone \bone\tp\right) \trans^{t} \errort(0).
    \label{eq:perperror}
\end{equation}
Recall Fact \ref{fact:similar}, which implies that   
$\vec x_1 = \inorm{D^{1/2}\bone}^{-1}D^{1/2}\bone$, and that we can write  
\[
    N = \inorm{D^{\frac 12}\bone}^{-2}D^{\frac 12}\bone (D^{\frac 12}\bone)\tp +
    \sum_{i=2}^{n} \eigtrans_i \vec x_i \vec x_i\tp. 
\]
Combining the previous observations with \eqref{eq:perperror}, we get
\begin{align}
        \norm{\errort\left(t\right)}
            &=\norm{\left(I - \frac{1}{n} \bone \bone\tp\right) \trans^{t}
                \errort(0)}\nonumber\\
            &=\norm{\left(I - \frac{1}{n} \bone \bone\tp\right) D^{-\frac 12} N^{t} D^{\frac 12}
                \errort(0)}\nonumber\\
            &\overset{(a)}{=}\norm{\left(I - \frac{1}{n} \bone \bone\tp\right) D^{-\frac 12}
                \left( \sum_{i=2}^{n} \eigtrans_i^{t} \vec x_i \vec
                x_{i}\tp \right) D^{\frac 12} \errort(0)}\nonumber\\
            &\overset{(b)}{\leq}\norm{\left(I - \frac{1}{n} \bone \bone\tp\right)}
                \norm {D^{-\frac 12}}
                \norm{ \sum_{i=2}^{n} \eigtrans_i^{t} \vec x_i
                \vec x_{i}\tp} \norm{D^{\frac 12}} \norm{\errort(0)}\nonumber\\
            &\overset{(c)}{\leq}
                \sqrt{\frac{\volmax}{\volmin}} \,
                \eigtrans_{*}^{t} \norm{\errort(0)}, 
            \label{eq:infamenonorthogonal}
\end{align}
where in $(a)$ we used the fact that 
$\left(I - \frac{1}{n} \bone \bone\tp\right) D^{-1/2} \vec x_1 = 0$, 
in $(b)$ we used the submultiplicativity of the norm and in $(c)$ we used that 
$\inorm{D^{-1/2}} \leq (\volmin)^{-1/2}$, that $ \inorm {D^{
1/2}} \leq (\volmax)^{1/2}$ and where by definition $\eigtrans_{*} = \max_{i \neq 1} 
|\eigtrans_i| = \max(\abs{\eigtrans_2},\abs{\eigtrans_n})$.

%If we rewrite $\errort\left(0\right)=\sum_{i=1}^{n} \alpha_i v_i $ where
%the $v_i$ are $D^{-1}A$'s eigenvectors, it follows that 
%\begin{align*}
%        \norm{\errort\left(t\right)}
%            &=\norm{\left(I - \frac{1}{n} \bone \bone\tp\right) (D^{-1} A)^{t}
%                \errort(0)}\\
%            &=\norm{\left(I - \frac{1}{n} \bone \bone\tp\right) (D^{-1} A)^{t}
%                \sum_{i=1}^{n} \alpha_i v_i }\\
%            &=\norm{\left(I - \frac{1}{n} \bone \bone\tp\right)
%                \sum_{i=1}^{n} \alpha_i \lambda_{i}^{t} v_i }\\
%            &\overset{(a)}{=}  \norm{\sum_{i=2}^{n}\left(I - \frac{1}{n} \bone \bone\tp\right)\alpha_i
%                \lambda_{i}^{t} v_i }\\
%            &\overset{(b)}{\leq}\norm{ \sum_{i=2}^{n} \alpha_i
%                \lambda_{i}^{t} v_i }
%                %\\ &
%            \overset{(c)}{\leq} \sum_{i=2}^{n} \norm{\alpha_i
%                \lambda_{i}^{t} v_i }
%                %\\ &
%            \overset{(d)}{=} \lambda^{t}\sum_{i=2}^{n} \norm{ \alpha_i v_i }
%            \overset{(e)}{\leq} \, \lambda^{t} n\norm{\errort\left(0\right)},
%\end{align*}
%where in $(a)$ we used that the final projection by $\left(I -
%\frac{1}{n} \bone \bone\tp\right)$ cancels the components of $(D^{-1}
%A)^{t}\errort(0)$ parallel to the first eigenvector $v_{1} = \bone$ of
%$D^{-1} A$, in $(b)$ the submultiplicativity of the spectral norm, 
%in $(c)$ the triangle's inequality, in $(d)$ we substituted
%$\lambda = \max_{i>1} |\lambda_i|$ and in $(e)$ the fact that each $|\alpha_i|$ 
%is at most $\norm{\errort\left(0\right)}$. 

\subsection{Time and message complexity}
In summary, the arguments in the previous subsection prove the following. 

\begin{theorem}
    \label{thm:jacobi-tc}
    After $t$ rounds, the orthogonal component of the error of the solution
    $\tilde{\vec p}(t)$ produced by Algorithm \ref{algo:jaco} is reduced by a
    factor
    \begin{equation}
        \frac{\norm{\errort(t)}}{\norm{\errort(0)}} \le 
                  {\left(\frac{\volmax}{\volmin}\right)}^{1/2} \,
                    \eigtrans_*^{t}, 
        \label{eq:jacobi-tc}
    \end{equation}
    where $\eigtrans_*$ is the second largest absolute value of an eigenvalue of
    $\trans$. The message complexity per round is $O(m)$. 
\end{theorem}
\begin{proof}
    The first claim follows by \eqref{eq:infamenonorthogonal}. For the second part
    of the claim, note that at any round of Algorithm \ref{algo:jaco}, each node
    sends its estimated potential value to each of its neighbors. Therefore, the
    number of exchanged messages is $O(m)$ per round. 
\end{proof}

Observe that in typical applications we have $\eigtrans_* = \eigtrans_2$.  
This is the case, for example, when one considers the \emph{lazy} variant of a
transition matrix in order to avoid pathological cases \cite[Section 1.3]{Levin:2009}. 
The condition $\eigtrans_* = \eigtrans_2$ implies, in particular, that \eqref{eq:jacobi-tc}
in Theorem \ref{thm:jacobi-tc} can be bounded in terms of the \emph{graph conductance}
of the network, 
    \footnote{We remark that here the term \emph{conductance} refers to the
    graph-theoretic notion also known as \emph{bottleneck ratio}
    \cite{Levin:2009}, and shall not be confused with the notion of
    \emph{electrical conductance} in the theory of electrical networks
    \cite{Doyle:1984}.}
since for any graph $\graph$, $\eigtrans_2(\graph) \le 1 - \phi(\graph)^2/2$  \cite[Theorem 13.14]{Levin:2009}.

\section{A token diffusion method}
\label{se:rw}
%In this section, we propose an estimator of the node potentials based on independent random walks. 
%In the rest of this section, for every $(u, 
%v)\in E$, we set $\w_{uv} = x_{uv}/\ell_{uv}$. The rest of this 
%section is organized as follows. 
%In Section \ref{subse:rw}, we define and justify our Monte Carlo based estimator of node potentials, while Section \ref{subse:high} discusses the stochastic quality of the estimation. 

%\subsection{A random walk perspective}\label{subse:rw}
Following a well-known analogy between electrical flows and random 
walks \cite{Doyle:1984,Tetali:1991,Chandra:1997,Levin:2009}, in this 
section we propose a random walk-based approach to approximate  
electric potentials. The process is described by Algorithm 
\ref{algo:rw_phys}. In each round, the algorithm starts $K$ new, 
mutually independent random walks at the source node. Each random 
walker (or \emph{token}) moves one step during each round of the 
algorithm, until it reaches the sink node, where it is absorbed. The 
independent parameter $K$ controls the accuracy of the process. 

\begin{algorithm}
	\DontPrintSemicolon % Some LaTeX compilers require you to use \dontprintsemicolon instead
    \SetKwFor{ProbDo}{with prob}{do}{}
    \SetKwInput{Parameter}{Parameters}
    \Parameter{$u \in \nodes$, $K \in \Nat$}
%	\KwIn{$u$, $K$}
	%\KwOut{The largest element in the set}
%	{\bf \{Step 1: Sending tokens\}}\;
	\tcp*[h]{Step 1: Send tokens}  \;
	\For{every token $T$ on $u$ \textbf{and} every neighbour $v$ of $u$} {
        \ProbDo{$\propto$ $w_{uv}$}{
            send $T$ to $v$ \;
            $Z(u) = Z(u) - 1$
        }
	}
	\tcp*[h]{Step 2: Receive tokens}  \;
	\For{every token $T$ received} {
		$Z(u) = Z(u) + 1$
	}
	\tcp*[h]{Step 3: Replenish source, deplete sink}  \;
	\If{$u=$ source} {
		%Inject $K$ new tokens at $u$\;
		${Z(u)} = Z(u) + K$  \tcp*{inject $K$ new tokens at $u$}
	}
	\If{$u=$ sink} {
		%Destroy all tokens at $u$\;
		$Z(u) = 0$   \tcp*{absorb all tokens at $u$}
	}
	\caption{Random walk algoritm.}
	\label{algo:rw_phys}
\end{algorithm}

\begin{algorithm}
	\DontPrintSemicolon % Some LaTeX compilers require you to use \dontprintsemicolon instead
	\KwIn{$u \in \nodes$, $K \in \Nat$}
	Return $\frac{\mathbf{Z}(u)}{K \vol(u)}$
	\caption{Estimator.}
	\label{algo:estim}
\end{algorithm}

Let $\ntok{t}{K}(u)$ denote the number of tokens at vertex $u$ at the 
end of round $t$, when $K$ new independent random walks are started at 
the source. Our \emph{estimator} of the potential at node $u$ at time $t$ will 
be 
\[
	\est{t}{K}(u)\defas \frac{\ntok{t}{K}(u)}{K\cdot \vol(u)}.
\]
%The estimator above is justified by the following result. 
We next show that {\em in expectation}, our estimator evolves following a recurrence that, 
though not identical, is very close to \eqref{eq:p_iterate}.

\begin{lemma}
    \label{le:estimator}
    Consider Algorithm \ref{algo:rw_phys} with $K = 1$. 
    Define inductively $\vpot{t} \in \Real^\nodes$ by 
    \begin{align}
        \pot{0}_u &= 0, \qquad \text{ for all $u \in \nodes$, } \\
        \pot{t+1}_u &= 
        \begin{cases}
            \frac{1}{\vol(u)} \left( \sum_{v \sim u} w_{uv} \pot{t}_v + b_u \right) 
                & \text{ if } u\neq \sink, \\
            0 
                & \text{ if } u=\sink. 
        \end{cases}
    \label{eq:grounded}
\end{align}

Then, for every time $t=0,1,2,\ldots$ and 
for every $u \in \nodes$ we have:
\begin{equation*}
	\ex[\est{t}{1}(u)]  = \pot{t}_u.
\end{equation*}
\end{lemma}
\begin{proof}
The claim is proved by induction. 
It clearly holds when $t=0$, since at that time there are no tokens and 
thus $\est{t}{1}=\vec 0=\pot{0}$. For $t \ge 1$, and for every $u \in 
\nodes \setminus \{sink\}$:
\begin{align*}
	\ex[\ntok{t+1}{1}(u)\mid\; \ntok{t}{1}]
    &= \sum_{v\backsim 
    u}\frac{\ntok{t}{1}(v)\w_{vu}}{\vol(v)} + b_u \\
    &= \sum_{v\backsim u}\w_{vu}\est{t}{1}(v) + b_u,
\end{align*}
where we used $\est{t}{1}(v) = \ntok{t}{1}(v)/\vol(v)$.
Dividing both sides by $\vol(u)$, we obtain:
\begin{align*}
	&\frac{\ex[\ntok{t+1}{1}(u)\mid\; \ntok{t}{1}] }{\vol(u)} = 
	\frac{1}{\vol(u)}\left (\sum_{v\backsim 
    u}\w_{vu}\est{t}{1}(v)+ b_u \right ).
\end{align*}
By recalling that $\est{t+1}{1}(u) = \ntok{t+1}{1}(u)/\vol(u)$ and by the law of
iterated expectations we obtain:
\begin{align}\label{eq:estim}
	\ex[\est{t+1}{1}(u)] &=
\begin{cases}	
	 \frac{1}{\vol(u)}\left (\sum_{v\backsim 
    u}\w_{vu}\ex[\est{t}{1}(v)] + b_u \right ) 
    & \text{ if } u \neq \sink, \\
    0 & \text{ if } u = \sink. 
\end{cases}
\end{align}
Recurrence \eqref{eq:estim} has the very same form as 
\eqref{eq:grounded}. 
In particular, from the inductive hypothesis $\ex[\est{t}{1}] = \vpot{t}$ we obtain $\ex[\est{t+1}{1}] = \vpot{t+1}$. 
This completes the proof.
\end{proof}

The following corollary 
justifies our estimator in Algorithm \ref{algo:estim}, when $K > 1$:

\begin{cor}\label{cor:estim}
Let
\[
	\est{t}{K}(u) \defas \frac{\ntok{t}{K}(u)}{K\vol(u)}, 
\]
for every $t=0,1,2,\ldots$ and $u \in \nodes$. Then:
\[
	\ex[\est{t}{K}(u)] = \pot{t}_u.
\]
\end{cor}
\begin{proof}
First of all observe that, obviously:
\begin{equation*}
\ex[\ntok{t}{K}(u)] = K \ex[\ntok{t}{1}(u)].
\end{equation*}
As a consequence:
\begin{equation*}
\frac{\ex[\ntok{t}{K}(u)]}{\vol(u)} = 
K\frac{\ex[\ntok{t}{1}(u)]}{\vol(u)} = K\pot{t}_u,
\end{equation*}
from Lemma \ref{le:estimator}. This proves the claim.
\end{proof}

Note that the definition of 
$\vpot{t}$ in Lemma \ref{le:estimator} is akin to that of 
$\tilde{\vec p}(t)$ in Equation \ref{eq:p_iterate}. One might thus 
reasonably expect that, like $\tilde{\vec p}(t)$, $\vpot{t}$ also 
converges to a solution of Kirchhoff's equations. Nevertheless, the two 
definitions {\em are} different and establishing this requires a separate 
proof, which we give in Section \ref{subse:diffu-apx}. 

That result will justify the interpretation of $\vpot{t}$, and hence of 
the vector $\vec V_K^{(t)}$, as an iterative approximation of the 
correct Kirchhoff potentials. Note that there are two sources of 
inaccuracy in this estimation. One is intrinsic to the iterative 
process, i.e., the rate with which $\vpot{t}$ converges to a solution 
of Kirchhoff's equations; this will be the subject of Section 
\ref{se:speed_tok}. The second source of error is stochastic and 
reflects the accuracy of the estimator itself; it will be discussed in 
Section \ref{subse:high}, where we show that for a large enough $K$, 
the estimator yields an accurate approximation of the potential with 
high probability, and not only in expectation. 

\subsection{Correctness of the token diffusion method}
\label{subse:diffu-apx}

We can reexpress the system \eqref{eq:grounded} as 
\begin{align}
    \begin{cases}
        \vpot{0} &= \vec 0, \\
        \vpot{t+1} &= \modW\, \vpot{t} + D^{-1} \modb,
    \end{cases}
    \label{eq:tokendiff}
\end{align}
where $\modW$ is obtained from $\trans \defas D^{-1} A$ by zeroing out all
entries on the row and column corresponding to the sink node. Likewise,
$\modb$ is obtained from $\vec b$ by zeroing out the entry corresponding to the
sink node. 
We next prove that the spectral radius of $\modW$ is \emph{strictly} between $0$ and $1$. Using this fact, we prove that the token diffusion method converges to a
feasible potential vector. 

\begin{lemma}
    \label{lem:specrad}
    The spectral radius of $\modW$, $\grho$, satisfies $0<\grho<1$. More precisely,
    $\grho = 1 - \sum_{i=1}^n v_i \cdot P_{i,\sink} /\norm[1]{\vec v} $, where
    $\vec v$ is the left Perron eigenvector of $\modW$. 
\end{lemma}
\begin{proof}
    First of all, observe that $\modW$ is diagonalizable: if we let $\gA$ be
    the matrix obtained from $A$ by zeroing out the entries on the row and
    column corresponding to the sink node, then $\modW = D^{-1} \gA$, and thus
    $\modW$ is similar to the symmetric real matrix $\modN = D^{-1/2} \gA D^{-1/2}$. 

    Moreover, $\modW$ is nonnegative and the Perron-Frobenius theorem for
    nonnegative matrices (for example, see  \cite[Section 8.3]{Meyer:2000})
    guarantees the existence of a nonnegative {\em row vector} $\vec v$ such that
    $ \vec v \modW = \grho \vec v$.
    Without loss of generality, assume that $\norm[1]{\vec v}=1$. 

    Now observe that 
    \begin{align*}
        \grho &= \abs{\grho} = \norm[1]{\grho \vec v} % \\ &
        = \norm[1]{ \vec v \modW } = \sum_{i=1}^n \sum_{j=1}^n v_i \modW_{ij}. 
    \end{align*}

    Let $\eps_{ij} \defas \trans_{ij} - \modW_{ij} \ge 0$ be the nonnegative
    ``gaps'' between $\trans$ and $\modW$. Then we can continue,
    \begin{align}
        \grho &= \sum_{i=1}^n \sum_{j=1}^n v_i \modW_{ij} 
        = \sum_i \sum_j v_i (\modW_{ij} + \eps_{ij} - \eps_{ij}) \nonumber\\
        &= \sum_i \sum_j v_i \trans_{ij} - \sum_i \sum_j v_i \eps_{ij} % \nonumber\\
        = \norm[1]{\vec v \trans } - \sum_i v_i \sum_j \eps_{ij} \nonumber\\
        &= 1 - \vec v \cdot \vec \gamma,\label{eq:specrad}
    \end{align}
    where $\vec \gamma$ is the vector of row gaps, i.e., $\gamma_i \defas \sum_j \eps_{ij}$, and for the last equality we also used the fact that every row of $\trans$ sums to 1. 
    Note that, for $i\neq \sink$, $\gamma_i$ is the same as $\trans_{i,\sink}$, that
    is, the probability that a token at node $i$ reaches the sink in one step.  

    We claim that the spectral radius is strictly positive. To see that, observe
    that, similarly to Fact \ref{fact:similar}, $\modW$ is
    similar to $\modN$ and thus shares the same eigenvalues. If the spectral
    radius $\grho$ were zero, $\modN$ would be a null matrix, i.e., 
    a matrix whose entries are all zeros
    (this follows by looking at the diagonalized form of $\modN$). 
    Since $\modW$ is similar to $\modN$, $\modW$ would be a null matrix as
    well, which is clearly not the case for any nonempty graph.     

    Observe also that $0 = \vec v \cdot \vec 0 = \vec v \modW \vec \chi_{\sink}  =
    \grho \vec v \cdot \vec \chi_{\sink}= \grho v_{\sink}$ and therefore $v_{\sink}=0$.
    To show that $\grho<1$, assume by contradiction that
    $\grho=1$. Then, $v_{\sink}=0$ and \eqref{eq:specrad} implies
    $\sum_{i\neq \sink} v_i \gamma_i =\sum_i v_i \gamma_i = 0$, i.e., $v_i=0$
    whenever $\gamma_i\neq 0$; in particular, $v_i=0$ for each node $i$
    adjacent to the sink. Continuing this argument would yield that $v_i=0$ for
    each $i$ adjacent to a node adjacent to the sink, (since the $i$th entry of
    $ \vec v \modW $ equals $\grho v_i$), and so on. Since the original
    graph is connected, this contradicts the fact that $\vec v \neq \vec 0$. 
\end{proof}

In the proof of next theorem, we make use of the following fact, which is analogous to Fact  \ref{fact:similar}.
\begin{fact}
    \label{fact:grundsimilar}
    The matrix $\modW$ is similar to
    the matrix $\modN$ obtained from $N$ by zeroing out its last column and row.
    In particular, $\modW^{t} = D^{-\frac 12} \modN^{t} D^{\frac 12}$, and
    $\modN$ has $n$ orthonormal eigenvectors $\modx_1, \dots, \modx_n$ which
    correspond to the eigenvalues $\grho_1 \ge \ldots \ge \grho_n$ of $\modW$.
\end{fact}

\begin{theorem}
    \label{thm:conv_of_grounded}
    The iterates of \eqref{eq:grounded} converge to a feasible solution of the
    linear system \eqref{eq:p_def}. The rate of convergence is proportional to 
    %the spectral radius $\grho$ of $\modW$. 
    $\grho$. 
\end{theorem}
\begin{proof}
    Since $\grho <1$, the matrix $I-\modW$ is invertible, and its inverse can be expressed as
    \[ (I - \modW)^{-1} = \sum_{k=0}^\infty \modW^k. \]
    %So we can express our desired solution as
    %$$ p^* = \sum_{k=0}^\infty P^k D^{-1} \chi_s. $$
    If we recursively expand the updates
    $ \vpot{t+1} = \modW \vpot{t} + D^{-1} \modb, $
    we get, for any $t \ge 1$,
    \begin{equation}
        \vpot{t} = \modW^t \vpot{0} + \sum_{k=0}^{t-1} \modW^k D^{-1} \modb =
        \left( \sum_{k=0}^{t-1} \modW^k \right) D^{-1} \modb,
        \label{eq:Pfinitesum}
    \end{equation}
    where we used that in \eqref{eq:tokendiff} $\vpot{0} = 0$.
    As $t \to \infty$, this yields
    \begin{equation}
        \vpot{\infty} \defas \lim_{t \to \infty} \vpot{t} = (I-\modW)^{-1} D^{-1} \modb, 
        \label{eq:Pinfinitesum}
    \end{equation}
    which shows that in the limit, the iterates satisfy the linear system 
    $ (I-\modW) \vpot{\infty} = D^{-1} \modb, $
    %or $\pot{\infty} = \modW \pot{\infty} + D^{-1} \modb$, 
    or, recalling that $\modW \defas D^{-1} \gA$,
    %\[ \pot{\infty} = \modW \pot{\infty} + D^{-1} \modb. \]
    %Since $\pot{\infty}(sink)=0$, we also have
    %\[  \]
    \begin{equation}
        (D - \gA) \vpot{\infty} = \modb.
        \label{eq:groundconv}
    \end{equation}
    To conclude that $\vpot{\infty}$ also satisfies the original system
    \eqref{eq:p_def}, notice the following. The two matrices $L = D-A$ and $D-\gA$,
    as well as the two vectors $\vec b$ and $\modb$, differ only in the row
    corresponding to the sink node, and the difference of the
    $sink$-th rows of the two matrices is given by the row vector
    \[ 
        \vec\chi_{\sink}\tp ((D - \gA) - L ) = \\
        \bone\tp (A - \gA ) = \bone\tp (D - \gA),
    \]
    where in the last equality we used $\bone\tp (D-A)=\vec 0\tp$. 
    Therefore, using \eqref{eq:groundconv}, 
    \begin{multline*}
        \vec\chi_{\sink} \tp 
        (D - \gA) \vpot{\infty} - \vec\chi_{\sink}\tp (D-A) \vpot{\infty} \\
        = \bone\tp(D-\gA) \vpot{\infty} 
        = \bone\tp \modb = 1 = \chi_{sink}\tp (\modb - \vec b),
    \end{multline*}
    which, after rearranging terms, together with \eqref{eq:groundconv} implies that 
    \[
        0 = \vec\chi_{\sink}\tp ((D - \gA) \vpot{\infty} - \modb) = \vec\chi_{\sink}\tp ((D- A)
        \vpot{\infty} - \vec b). 
    \]
    This proves that $(D-A)\vpot{\infty}= L\vpot{\infty}=\vec b$. 

    The second part of the theorem follows from Fact \ref{fact:grundsimilar},
    \eqref{eq:Pfinitesum} and \eqref{eq:Pinfinitesum}, which yield
    \begin{multline*}
        \|\vpot t - \vpot \infty\| 
        = \left\| \left( \sum_{k=t}^{\infty} \modW^k \right) D^{-1} \modb \right\|
        \leq \sum_{k=t}^{\infty} \left\| \modW^k  D^{-1} \modb \right\|\\
        \stackrel{(a)}{ \leq } \sqrt{\frac{\volmax}{\volmin}} 
            \frac {\sum_{k=t}^{\infty} \grho^k}{\vol(source)} 
        =\sqrt{\frac{\volmax}{\volmin}} 
        \frac { \grho^t}{(1-\grho)} \frac 1{\vol(source) },
    \end{multline*}
    where in $(a)$ we performed a calculation analogous to \eqref{eq:infamenonorthogonal}.
\end{proof}

\subsection{Convergence rate of token diffusion}
\label{se:speed_tok}

In Section \ref{subse:diffu-apx} we showed that the token diffusion system  
converges to one of the solutions of Kirchhoff's equations. Moreover, 
its rate of convergence is dictated by the spectral radius of the
transition matrix $\modW = D^{-1} \gA$ (Theorem 
\ref{thm:conv_of_grounded}), which is similar to the original transition matrix $\trans$, except for the fact
that all entries of the row and column corresponding to the sink are equal to
$0$ in $\modW$. 
%Grounding the sink amounts to zeroing all entries of the row and column 
%corresponding to the sink in the adjacency matrix, or any matrix 
%derived from it, such as the (possibly normalized) Laplacian. 

For simplicity of exposition, in the remainder we simply remove the 
row and column of $\modW$ corresponding to the sink. Assume without loss of 
generality that the sink corresponds to the $n$-th row/column index. 
Given an $n\times n$ matrix $M$, consider the $(n-1)\times(n-1)$ matrix 
$\grM$, obtained from $M$ by {\em grounding} the $n$-th index, that is, 
removing the $n$-th row and $n$-th column. The next fact shows that 
this operation does not affect the spectral radius of $\modW$.

\begin{proposition}\label{fa:grounded}
Assume $M$ is an $n \times n$ matrix where each entry of the $n$-th row and of the $n$-th column is zero. Then: 
\begin{enumerate}
\item for every eigenpair $((x_1,\ldots,x_{n-1}), \mu)$ of $\grM$ there is an eigenpair $((x_1,\ldots,x_{n-1},0), \mu)$ of $M$; 
\item $(\chi_n, 0)$ is an eigenpair of $M$; 
%\item for every other eigenpair $((x_1,\ldots,x_n), \mu)$ of $M$ there is an eigenpair $((x_1,\ldots, x_{n-1}), \mu)$ of $\underline{M}$;
\item the spectral radius of $M$ and $\grM$ is the same. 
\end{enumerate}
\end{proposition}
\begin{proof}
Point (1) follows from the assumption that the $n$-th column of $M$ is
identically zero and thus the $i$-th entry of $M (x_1, \ldots,$ $ x_{n-1}, 0)$ is
equal to the $i$-th entry of $\grM (x_1,\ldots,x_{n-1})$ for any
$i<n$. 
Point (2) follows from $M \chi_n = \vec 0$. 
%For point (2), observe that since $M_{nj}=0$  for every $j=1,\ldots,n$, from the assumption we have $\mu x_n = (M \vec x)_n = 0$, so either $\mu=0$ or $x_n=0$; in the latter case, $(x_1,\ldots,x_{n-1})$ must be an eigenvector of $\underline{M}$ with eigenvalue $\mu$. 
Point (3) is a direct consequence of the first two. 
\end{proof}

Since $\modW$ satisfies the hypothesis of Proposition \ref{fa:grounded}, we can
equivalently study the spectral radius of $ \grmodP$ $( = \grP)$. 
To simplify (and with a slight abuse of) notation, in the remainder of
this section we write $\underline{M}$ for $\mathrm{gr}(M)$.  

\paragraph{Additional notation.}
We denote by $\modG$ the graph obtained from $\graph$ by removing the sink node and its incident edges. 
We denote by $\glap$ the Laplacian matrix of $\modG$, so that $\glap = \gD - \modA$, with 
$\modA$ and $\gD$ respectively the adjacency and degree matrices of $\modG$. 
We also define $\modlap \defas \modD - \gA$. 
Note that $(\modA)_{ij} = (\underline{A})_{ij}$ for each $i,j<n$. 
On the other hand $\modlap$ is not a proper graph Laplacian, since 
$(\gD)_{ii} \ne (\modD)_{ii}$ for some $i<n$.\footnote{Precisely, this 
happens whenever $i\sim n$ in the original graph $\graph$.}
However, $\modlap$ can be viewed as a \emph{perturbed} Laplacian, since $\modD = \gD + \smat$, where $\smat = \text{diag}(\w_{1,n},\ldots , \w_{n-1,n})$. 

The rate of convergence of the token diffusion process is dictated by
$\underline{\rho}$, the dominant eigenvalue of the matrix $\modW \defas D^{-1} \gA =
\modD^{-1} \gA$.
Thanks to Fact \ref{fact:grundsimilar}, we can equivalently study the matrix
$\modN = \modD^{-1/2}\,\gA\,\modD^{-1/2}$, which shares the same spectrum as
$\modW$, or, equivalently, the matrix $\nglap = I - \modN =
\modD^{-1/2}\modlap\modD^{-1/2} = \modD^{-1/2}(\glap + \smat)\modD^{-1/2}$.
Again, the matrix $\nglap$ can be interpreted as a perturbed normalized Laplacian. 
%(indeed, we know from the previous section that the spectral radius of $\modN$ is strictly less than $1$).
The eigenvalue $\grho$ of $\modN$ corresponds to an eigenvalue $\glam =
1-\grho$ of $\nglap$.
%
%Note also that, by Lemma \ref{lem:specrad}, $\eigr$ is the smallest eigenvalue of  $\nglap$. 

\medskip
In this section we provide a lower bound on $\eigr$. Let 
$\bx$ denote the (unit norm) eigenvector of $\nglap$ corresponding to 
$\eigr$ and let $\by = \modD^{-1/2}\bx$. Since $\nglap$ is symmetric, we have by definition:
\begin{equation}\label{eq:gen_exp_eigr}
\begin{split}
	\eigr = \bx\tp\nglap\bx = \bx\tp\modD^{-1/2}\modlap\modD^{-1/2}\bx = \\
    \bx\tp\modD^{-1/2}(\glap + \smat)\modD^{-1/2}\bx = \by\tp(\glap + 
	\smat)\by,
\end{split}
\end{equation}
%Note that if $\modlap \defas \glap + \smat$ were
%a proper Laplacian we would have $\by = \bone$. In light of this, we express

\begin{proposition}\label{fa:norm_y}
The following holds:
\begin{equation}\label{eq:norm_y}
	\frac 1{\vol_{\max}} \le\|\by\|^2\le \frac 1{\vol_{\min}}.
\end{equation}
%where $\vol_{\max}=\max_i \vol(i)$, $\vol_{\min}=\min_i \vol(i)$. 
\end{proposition}
\begin{proof}
We have:
\[
	\|\by\|^2 = \|\modD^{-1/2}\bx\|^2 = 
    \sum_{i=1}^{n-1} \frac{\bx_i^2}{\vol(i)}.
\]
The claim then follows immediately since $\|\bx\|^2 = 1$.
\end{proof}
In the remainder, we decompose $\by$ as $\by = \bu + \bz$, with $\bu$ 
and $\bz$ the components of $\by$ respectively parallel and orthogonal 
to the vector $\bone$. The next fact highlights a general property of the perturbed 
Laplacian matrix that affords a simplification of \eqref{eq:gen_exp_eigr}. 
\begin{proposition}\label{fa:gen_decomposition}
%Consider any vector $\by \in \Real^{n-1}$ and let $\by = \bu + \bz$, with $\bu$ and $\bz$ respectively its components parallel and orthogonal to $\bone$. 
For any $\by \in \Real^{n-1}$, it holds $ \by\tp \glap \by = \bz\tp\glap\bz$, where
$\bz = \by - ((\bone\tp \by)/(\bone\tp \bone)) \cdot \bone$ is the component of $\by$ orthogonal to $\bone$.
%We have:
%\begin{equation}\label{eq:gen_decomposition}
%	\by\tp(\glap + 	\smat)\by = \by\tp\smat\by + \bz\tp\glap\bz.
%\end{equation}
%\qed
\end{proposition}
\begin{proof}
    If $\bu \defas \by - \bz$, note that $\bu$ is parallel to $\bone$, so
%    \begin{align*}
%        \by\tp(\glap + \smat)\by 
%        & = (\bu + \bz)\tp(\glap + \smat)(\bu + \bz)\\
%        & \stackrel{(a)}{ = } \bu\tp\smat\bu + \bu\tp\smat\bz + \bz\tp\smat\bu + \bz\tp\smat\bz + 
%        \bz\tp\glap\bz\\
%        & = \bu\tp\smat(\bu + \bz) + \bz\tp\smat(\bu + \bz) + \bz\tp\glap\bz = 
%        \by\tp\smat\by + \bz\tp\glap\bz,
%    \end{align*}
%    where $(a)$ follows since $\bu$ is in the kernel of $\glap$ by 
%    construction.\footnote{Recall that $\glap$ is the Laplacian of 
%    $\modG$.}
    \begin{align*}
        \by\tp \glap  \by 
        & = (\bu + \bz)\tp \glap (\bu + \bz)\\
        & = \bu\tp\glap\bu + \bu\tp\glap\bz + \bz\tp\glap\bu + \bz\tp\glap\bz \\
        & \stackrel{(a)}{ = } \bz\tp\glap\bz,
    \end{align*}
    where $(a)$ follows since $\bu$ is in the kernel of $\glap$ by 
    construction.%\footnote{Recall that $\glap$ is the Laplacian of    $\modG$.}
\end{proof}
 
We can now give a lower bound on $\eigr$, in terms of $\by$ and $\bz$. 
\begin{lemma}\label{le:lb_gen_eigr}
The following holds:
\begin{equation}\label{eq:lb_gen_eigr}
	\eigr\ge\by\tp\smat\by + \reig_2\|\bz\|^2,
\end{equation}
where $\reig_2$ is the second smallest eigenvalue of $\glap$, the  
Laplacian of the graph $\modG$.
\end{lemma}
\begin{proof}
    From \eqref{eq:gen_exp_eigr} and Proposition \ref{fa:gen_decomposition}, we have
    \begin{align*}
        &\eigr = \by\tp(\glap + 	\smat)\by = \by\tp\smat\by + \bz\tp\glap\bz = 
        \by\tp\smat\by + \frac{\bz\tp\glap\bz}{\|\bz\|^2}\|\bz\|^2\\
        &\ge\by\tp\smat\by + \|\bz\|^2\min_{\stackrel{\bw\perp\bone}{\|\bw\| 
        = 1}}\bw\tp\glap\bw = \by\tp\smat\by + \reig_2\|\bz\|^2,
    \end{align*}
    where the inequality follows by recalling that $\bz\perp\bone$ by 
    definition, and by observing that the second term of the sum is 
    the Rayleigh quotient associated to $\glap$, multiplied by $\|\bz^2\|$.
\end{proof}

We can now prove the main result of this section.
\begin{theorem}\label{thm:min_lambda}
    The following holds:
    \begin{equation*}
        \eigr\ge 
        \frac{\reig_2}{2\vol_{\max}(n - 1)}\sum_i\frac{\w_{in}}{\w_{in} + 
        \reig_2}.  
    %	\eigr\ge\min\left\{\frac{\reig_2}{2\vol_{\max}}, 
    %	\frac{\reig_2}{2\vol_{\max}(n - 1)}\sum_i\frac{\w_{in}}{\w_{in} + 
    %	\reig_2}\right\},
    \end{equation*}
\end{theorem}
\begin{proof}
    Since $\bu$ is the component of $\by$ parallel to $\bone$, we have $\bu_i^2 = 
    u^2/(n - 1)$ for some $u \in \Real$.
    Given $\bu$, we compute the vector $\bz$ that minimizes the 
    right-hand side of \eqref{eq:lb_gen_eigr}. Let $f(\bz) = (\bu + 
    \bz)\tp\smat(\bu + \bz) + \reig_2\|\bz\|^2$, where $\bu$ is regarded as a constant. We have:
    \begin{align*}
        &\frac{\partial f}{\partial\bz_i} = 2\w_{in}(\bu_i + \bz_i) + 
        2\reig_2\bz_i,\\
        &\frac{\partial^2f}{\partial\bz_i\partial\bz_j} = 0,\, i\ne j,\\
        &\frac{\partial^2f}{\partial\bz_i^2} = 2\w_{in} + 2\reig_2.
    \end{align*}
    Since $\w_{in}, \reig_2 > 0$ (the latter following since $\reig_2$ is 
    the second eigenvalue of a Laplacian matrix), the determinant of the 
    Hessian matrix is positive, hence $f(\bz)$ has a global minimum that is 
    the critical point. If we set the $i$-th first-order partial derivative 
    to $0$ we obtain
    %\[
    $
        \bz_i = - \w_{in}\bu_i / (\w_{in} + \reig_2).
    $
    %\]
    Substituting back into $f(\bz)$ yields:
    \begin{align}
        &f(\bz) = \sum_i\w_{in}\left(\bu_i - \frac{\w_{in}\bu_i}{\w_{in} + 
        \reig_2}\right)^2 + \reig_2\sum_i\left(\frac{\w_{in}\bu_i}{\w_{in} 
        + \reig_2}\right)^2\nonumber\\
        & = \reig_2^2\sum_i\frac{\w_{in}\bu_i^2}{(\w_{in} + \reig_2)^2} + 
        \reig_2\sum_i\left(\frac{\w_{in}\bu_i}{\w_{in} + 
        \reig_2}\right)^2\nonumber\\
        &= \frac{\reig_2u^2}{n - 1}\sum_i\frac{\w_{in}}{\w_{in} + 
        \reig_2}\label{eq:lb}
    \end{align}
    Next, recall that $\|\bu\|^2 + \|\bz\|^2\ge {1/\vol_{\max}}$ from Proposition 
    \ref{fa:norm_y}. We consider two cases. If 
    $\|\bz\|^2\ge 1 /(2\vol_{\max})$, then we have:
    \begin{equation}
    \label{eq:bound1}
        \eigr\ge\frac{\reig_2}{2\vol_{\max}}.
    \end{equation}
    Otherwise, $u^2\ge 1 / (2\vol_{\max})$, and \eqref{eq:lb} immediately gives
    \begin{equation}
    \label{eq:bound2}
        \eigr\ge\frac{\reig_2}{2\vol_{\max}(n - 1)}\sum_i\frac{\w_{in}}{\w_{in} + 
        \reig_2}.
    \end{equation}
    To conclude the proof, note that the latter bound is always the worse case, since $w_{in}/(w_{in}+\reig_2) \le 1$.
\end{proof}

\subsection{Stochastic accuracy of the estimator}
\label{subse:high}
In this subsection, we study how closely the estimator $\est{t}{K}(u)$ 
approximates its expected value, the potential $\pot{t}_u$. In a 
nutshell, we show the following: i) the larger $K$ (the 
independent parameter controlling the number of tokens injected per 
round), the higher the accuracy; ii) all the rest being equal, the 
higher the potential, the higher the accuracy.

\paragraph{Notation.} Starting at 
$t = 0$, we index tokens in increasing order of their release dates. In 
more detail, any token released {in the $(i+1)$-th round has an index in the 
interval $\{iK + 1,\ldots , (i+1)K\}$, with $i = 0, 1,\ldots$}, while the 
relative order of tokens released in the same round is irrelevant and 
arbitrary. 
The main result of this subsection is the following. 
\begin{theorem}\label{thm:high}
    For any given $K$, $0 < \epsilon, \delta < 1$, for every $t$ and for 
    every $u$, such that 
    $\pot{t}_u\ge\frac{3}{\epsilon^2K\vol(u)}\ln\frac{2}{\delta}$,
    Algorithms \ref{algo:rw_phys} and \ref{algo:estim} together provide an 
    $(\epsilon, \delta)$-approximation of $\pot{t}_u$.\footnote{A random 
    variable $X$ gives an $(\epsilon, \delta)$-approximation of a 
    non-negative quantity $Y$ if $\Prob{}{|X - Y| > \epsilon Y}\le\delta$.}
\end{theorem}
\begin{remark}
    Given the statistical and node-wise nature of the counter estimator, there
    is a ``resolution'' limit for the minimum value of a potential that can be
    estimated with desired accuracy and confidence levels for a specific
    value of $K$. 
    This is a consequence of the law of large numbers (applied in
    the form of a Chernoff bound in our case). A similar issue would arise if
    we used a different estimator, e.g., one based on Tetali and Snell's approach.
    On the other hand, accuracy and confidence can be improved by increasing
    $K$.
    This leads to an equivalent way of
    expressing Theorem \ref{thm:high} in which, given the minimum potential
    value we want to estimate with given accuracy and confidence levels, we can
    compute the minimum $K$ that achieves the desired performance. More
    formally, an $(\epsilon, \delta)$-approximation of the potentials $\pot{t}_u$ greater
    than $\pot{t}_{\star}$ can be
    achieved by setting
    $K\ge\frac{3}{\epsilon^2 \pot{t}_{\star} \vol(u)}\ln\frac{2}{\delta}$. 
\end{remark}
\begin{proof}[Proof of Theorem \ref{thm:high}]
    Let $\isthere{t}_j(u) = 1$ if the $j$-th token is at 
    node $u$ at time $t$, $\isthere{t}_j(u) = 0$ otherwise. 
    From Corollary \ref{cor:estim} %and the definition of $\isthere{t}_i(u)$, 
    we have:
    \begin{align*}
        K\, \vol(u) \, \pot{t}_u &= \ex[\ntok{t}{K}(u)] = 
        { \ex[\sum_{j=1}^{Kt} \isthere{t}_j(u)]}. 
    \end{align*}

    The $\isthere{t}(u)$'s are independent Bernoulli variables and the 
    expectation of their sum is $K\vol(u)\pot{t}_u$. Hence, a simple 
    application of the multiplicative Chernoff bound yields
    \begin{equation*}
        \begin{split}
            \Pr\left[\abs{\ntok{t}{K} - K\vol(u)\pot{t}(u)} > \epsilon 
            K\vol(u)\pot{t}(u)\right]\le \\ 
            \le 2e^{-\frac{\epsilon^2}{3}K\vol(u)\pot{t}(u)}\le\delta,
        \end{split}
    \end{equation*}
    whenever 
    $\pot{t}_u\ge\frac{3}{\epsilon^2K\vol(u)}\ln\frac{2}{\delta}$. 
    Finally, note that 
    \begin{equation*}
        \abs{\ntok{t}{K} - K\vol(u)\pot{t}_u} \le \epsilon 
        K\vol(u)\pot{t}_u\iff \abs{\est{t}{K}(u) - \pot{t}_u} \le \epsilon 
        \pot{t}_u.
    \end{equation*}
    directly from definitions. This completes the proof.
\end{proof}

\subsection{Time and message complexity}
%\vin{TBD : Recap here the time complexity.}
The arguments from previous sections lead to the following conclusions about
the token diffusion process. 

\begin{theorem}
	The expected value of the estimator vector $\vec V_K^{(t)}$ 
	constructed by Algorithms \ref{algo:rw_phys} and \ref{algo:estim} 
	converges to the grounded solution $\vec p$ of the Kirchhoff 
	equations at a rate 
    \begin{equation}
        \label{eq:l2bound}
        O((1-\eigr)^t) = O\left(\left( 1 - \frac{\reig_2}{2n \,
            \vol_{\max}}\sum_{u\in \nodes, u\neq n}\frac{\w_{un}}{\w_{un} +
            \reig_2} \right)^t\right).
    \end{equation}
\end{theorem}
\begin{proof}
    By Corollary \ref{cor:estim}, the expected value of the estimator vector at
    time $t$ is given by the vector $\vpot{t}$. 
    From the analysis in Section \ref{subse:diffu-apx}, we know that the rate
    of convergence of $\vpot{t}$ to $\vec p$ is dictated by $1-\eigr$, the
    spectral radius of the matrix $\modW$. The result thus follows from Theorem
    \ref{thm:min_lambda}. 
\end{proof}
Note that the right hand side in \eqref{eq:l2bound} is \emph{decreasing} with
$\reig_2$. Thus, any lower bound on $\reig_2$ yields an upper bound on the
right hand side in \eqref{eq:l2bound}. 
By recalling that $\reig_2$ is the second smallest eigenvalue of the graph
$\modG$, this allows to connect the error term in \eqref{eq:l2bound} to the
\emph{edge expansion} of $\modG$, since for any graph $\graph$ it is known
\cite[Theorem 2.2]{Berman:2000} that
\[
    \lambda_2(\graph) \ge \volmax - (\volmax^2 - \theta(\graph)^2)^{1/2}. 
\]
Thus, the higher the edge expansion of $\modG$, the higher $\reig_2$, and the
faster the convergence of the token diffusion process. 

We finally derive a bound on the expected message complexity of Algorithm \ref{algo:rw_phys}. 
\begin{proposition}
    As $t \to \infty$, the expected message complexity per round of Algorithm \ref{algo:rw_phys} is 
    $
    O(K \, n \, \vol_{\mathrm{max}} \cdot \energy)
    $,
    where $\energy \defas \vec p\tp L \vec p$ is the \emph{energy} of the electrical flow. 
\end{proposition}
\begin{proof}
    The expected message complexity per round is given by the expected number
    of tokens moved by the diffusion algorithm in one step. In the regime $t
    \to \infty$, $\vpot{t}$ approaches the Kirchhoff potentials $\vec p$. By
    Corollary \ref{cor:estim}, the expected number of tokens at a node $u\in
    \nodes$ approaches
    %\[
    $
        \ex[\ntok{\infty}{K}(u)] = K \, \vol(u) \, \pot{\infty}_u = K \, \vol(u) \, p_u,
    $
    %\]
    thus the expected total number of tokens is 
    \[
        \sum_{u \in \nodes} \ex[\ntok{\infty}{K}(u)] = K \sum_u \vol(u) p_u \le
        K n \vol_{\mathrm{max}} p_{\source}, 
    \]
    where the inequality follows from the fact that $0 = p_{\sink} \le p_u \le p_{\source}$. Observing that 
    %\[
    $
        p_{\source} = p_{\source} - p_{\sink} = \vec p\tp \vec b = \vec p\tp L \vec
        p = \energy
    $
    %\]
    concludes the proof. 
\end{proof}

%\vin{The bound above is not very explicit and also, is not in a normalized form. However as a corollary it easily yields a cruder bound of
%\[ 
%O\left( K (n-1) \delta_{\mathrm{max}} \mathrm{diam}(\graph) \frac{w_{\mathrm{max}}}{w_{\mathrm{min}}} \right),
%\]
%where $\delta_{\mathrm{max}}$ is the maximum degree, which is explicit and normalized. (To see why, use the fact that the energy is at most the length of any path with weights given by resistances $w_e^{-1}$, and use $\volmax \le \delta_{\mathrm{max}} w_{\mathrm{max}}$.) Perhaps a stronger normalized bound can be obtained by a more direct analysis. }
%

\section{Beyond potentials: An outlook}
\label{se:conclu}
%We proposed and analyzed two processes that can be used to estimate the 
%node potentials of a resistive network with a source and a sink. %We 
%have proved that both processes converge to a valid solution of 
%Kirchhoff's equations. For the process based on Jacobi's iterative 
%method, we have shown that the rate of convergence can be bounded in 
%terms of the graph conductance of the underlying network. For the 
%process based on random walkers, the rate of convergence has been 
%bounded in terms of the edge expansion of the network with the sink 
%node removed. We also derived bounds for the message complexity of 
%these methods. 

Our results show that the effectiveness of decentralized, simple 
processes for electrical flow computation can be quantitatively 
analyzed, which is a step forward in the microscopic-level analysis of 
social, biological and artificial systems that can be described 
in terms of time-varying resistive networks or as current-reinforced 
random walks. 

On the other hand, while opinion dynamics were originally proposed 
as elementary models of information exchange and manipulation in social 
networks, our results highlight their potential as versatile and powerful 
primitives for collective computing. We believe this is a perspective that 
deserves further investigation.

%%%%%%%%%%%%%%%%%%%%%%%%%%%%%%%%%%%%%%%%%%%%%%%%%%%%%%%%%%%%%%%%%%%%%%%%%%%%%%%%%%%%%%%%%%%%%%%%%%%%%%%%%
%% bibliography: see CFP for number of permitted pages

%\bibliographystyle{plain}  % do not change this line!
%\bibliography{bib}

\end{document}